\documentclass[11pt]{article}
\usepackage[T1]{fontenc}
\usepackage[margin=1in]{geometry}
\usepackage{amsmath,amssymb,amsthm}
\usepackage{graphicx}
\IfFileExists{cite.sty}{\usepackage{cite}}{} 

\newtheorem{theorem}{Theorem}[section]

\newtheorem{lemma}{Lemma}[section]
\theoremstyle{remark}
\newtheorem{remark}{Remark}[section]
\usepackage{setspace}  
\usepackage{authblk}   % preamble

\title{Lee-Yang Zeros And Particle Fluctuations}

\author[1]{M.E.H. Bahri\thanks{\texttt{mbahri@math.rutgers.edu}}}
\author[1]{Ian Jauslin\thanks{\texttt{ian.jauslin@rutgers.edu}}}
\author[1,2]{Joel L. Lebowitz\thanks{\texttt{lebowitz@math.rutgers.edu}}}
\affil[1]{Department of Mathematics, Rutgers University, 08854, U.S.A.}
\affil[2]{Department of Physics and Astronomy, Rutgers University, 08854, U.S.A.}

\date{}

\begin{document}
\maketitle

\begin{abstract}We consider classical particles in the continuum in the grand canonical ensemble, with a stable, tempered and lower-regular pair potential and boundary conditions of uniformly bounded density. We prove that if the Lee--Yang zeros of the grand canonical partition function in the complex fugacity plane $z = e^{\beta\mu}$ remain bounded away from a real point $z_0 > 0$ for all sufficiently large volumes, then along cubes the thermodynamic limit and differentiation commute at $z_0$: every derivative of the finite-volume pressure in the chemical potential converges, uniformly in a neighborhood of $z_0$, to the corresponding derivative of the limiting pressure. The limiting values of all derivatives are independent of the boundary condition; in particular, the density and the particle-number variance per unit volume converge to $\beta^{-1}\partial_\mu p$ and $\beta^{-2}\partial^{2}_{\mu} p$, respectively. The result extends to the unbounded boundary conditions of Procacci and Yuhjtman for super-stable potentials in addition to Ruelle's tempered boundary conditions.\end{abstract}

\section{Introduction}

Many properties of a system depend on the boundary condition imposed on a finite region $\Lambda$. The pair-correlation function is one property: distinct infinite-volume Gibbs measures, selected by different boundary conditions, carry different correlation functions, and coexisting phases differ in all their local properties \cite{ruelle1969statistical,georgii2011gibbs}. What effect boundary conditions have on measurable macroscopic quantities, such as the compressibility, is the main subject of this study.

The finite-volume compressibility, which is governed by the particle-number variance $\sigma_\Lambda = \langle N_\Lambda^2\rangle - \langle N_\Lambda\rangle^2$, is of particular interest as the thermodynamic limit is taken. In part, this is due to its historical significance in detecting phase transitions. For example, at the liquid--gas critical point, the compressibility diverges due to strong density fluctuations, thus leading to critical opalescence \cite{einstein1910theorie}. The compressibility is also relevant away from points of phase transitions, particularly in detecting anomalous properties of phases. Certain systems possess anomalous non-extensive fluctuations, where $\sigma_\Lambda$ grows more slowly than the volume $|\Lambda|$---a property called hyperuniformity, or super-homogeneity \cite{torquato2003local,torquato2018hyperuniform}. This leads to a vanishing compressibility in the thermodynamic limit. The most elementary examples are at $T = 0$, where periodic crystals and many quasicrystals possess number fluctuations that grow only as the surface area, $|\Lambda|^{\frac{d-1}{d}}$ \cite{torquato2003local,torquato2018hyperuniform}. At positive temperatures, one of the most prominent examples of hyper-uniformity is the classical Coulomb system, where screening sum rules suppress extensive charge fluctuations, which again grow only as the surface area \cite{martin1980charge,lebowitz1983charge}. One can therefore observe that hyperuniformity does not require long-range order: disordered hyperuniform states, statistically isotropic and without Bragg peaks, were first found in maximally random jammed sphere packings \cite{donev2005unexpected} and now appear across condensed-matter, mathematical, and biological systems \cite{torquato2018hyperuniform}.

In fact, equilibrium systems with stable, tempered interactions are generally expected not to be hyperuniform. Ginibre proved a lower bound on the compressibility of a system in a finite volume $\Lambda$ with a pair potential $\phi \ge 0$ or a potential with a hard-core \cite{ginibre1967rigorous}, so the fluctuations are extensive; mathematically, this amounted to a lower bound on the second derivative of the finite-volume pressure $p(\mu;\Lambda,b)$ in $\mu$. And therefore harmonic crystals at positive temperature are not hyper-uniform, contrary to $T = 0$. In addition, with a careful and detailed revision, one may readily see that Ginibre's arguments, which were done with free boundaries, could be carried out for arbitrary boundary conditions. Therefore, a uniform lower bound on the compressibility could be extended independently of the boundary condition.

Dereudre and Flimmel extended Ginibre's results, proving non-hyperuniformity for a broad class of Gibbs point processes, including those with a super-stable, lower-regular, integrable pair potential, hence beyond the purely repulsive case \cite{dereudre2024nonhyperuniformity}. They work with infinite-volume Gibbs measures, so their lower bound holds for every Gibbs measure (defined below) and is independent of the boundary condition.

Despite these remarkably general lower bounds on the number variance, there remains the question of whether the limiting value per unit volume is the same for different boundary conditions: two co-existing phases could both be non-hyperuniform yet have different compressibilities. Nor is it well understood how the compressibility of a finite box $\Lambda$ in an infinite system compares, as $|\Lambda| \to \infty$, with the infinite-volume compressibility, or whether that limit is common to every Gibbs measure at a given fugacity.

To state our main result, we must introduce a few notions. We consider the consistent family of all finite-volume measures, over all bounded $\Lambda$ and all boundary configurations, known as the Gibbs specification of the interaction. From this specification we select a finite region, $\Lambda$, within an infinite system. To describe the local properties within $\Lambda$ we must then fix a boundary condition $b$, which may be periodic or a specified configuration of particles in $\Lambda^{c}$. Fixing such a configuration determines a probability distribution on the configurations inside $\Lambda$. Since an arbitrary number of particles may be in $\Lambda$, allowing for density and energy fluctuations, the grand-canonical (GC) ensemble is the correct framework.

An infinite-volume measure is a Gibbs measure, and satisfies the Dobrushin--Lanford--Ruelle (DLR) equations, if for every bounded $\Lambda$ its conditional distribution inside $\Lambda$, given the configuration in $\Lambda^{c}$, coincides with the finite-volume measure carrying that configuration as boundary condition \cite{georgii2011gibbs,dobrushin1968problem,lanford1969observables}. The Gibbs measures form a simplex whose extreme points are the pure phases, and more than one exists exactly when phases coexist \cite{georgii2011gibbs}; the pressure is common to all of them, but the measures and their correlation functions need not be. The GC partition function
\begin{equation}\label{eq:partfun}
\Xi(\mu;\Lambda,b) = \sum_{N=0}^{\infty} \frac{e^{\beta\mu N}}{N!} \int\!\cdots\!\int_{\Lambda^{N}} e^{-\beta\phi(x_1,\dots,x_N)} \prod_{i=1}^{N} e^{-\beta\phi_b(x_i)}\, dx_i
\end{equation}
with $\beta = (k_B T)^{-1}$ and $N$ is the number of particles in the volume. The pair potential $\phi(x_1,\dots,x_N) = \sum_{i<j} \phi(|x_i - x_j|)$ where $\phi_b(x_i)$ is the energy of the $i$-th particle due to the boundary configuration. Since all the coefficients of $z\equiv e^{\beta\mu}$ in Eq.~\eqref{eq:partfun} are non-negative, we have $\Xi(z;\Lambda,b) \ge 1$ for $z > 0$ and it is monotonically increasing with $z$. We assume henceforth that $\phi$ is stable, tempered and lower-regular in the sense of Ruelle \cite{ruelle1969statistical}:
\begin{equation}\label{eq:s}
\text{stable:}\quad \exists B \ge 0 \qquad \phi(x_1,\dots,x_N) \ \ge\ - B\, N
%\text{stable:}\quad \exists B \ge 0 \qquad \phi(x_1,\dots,x_N) \ \ge\ \sum_{q\in\mathbb{Z}^{d}}  - B\, N_q
\end{equation}
for every finite configuration.
\begin{align}
&\text{tempered:}\quad \exists\, R_0 \ge 0,\ \varepsilon > 0 \ \text{ such that } \ \phi(x) \le |x|^{-(d+\varepsilon)} \quad \forall\, |x| \ge R_0, \label{eq:temp}\\
&\text{lower-regular:}\quad \exists\, R_1 \ge 0 \ \text{ and a decreasing } \psi\colon [R_1,\infty)\to[0,\infty) \ \text{ such that} \label{eq:lr1}\\
&\qquad\qquad \phi(x) \ \ge\ -\,\psi(|x|) \quad \forall\, |x| \ge R_1, \notag\\
&\qquad\qquad \int_{R_1}^{\infty} \psi(r)\, r^{d-1}\, dr < \infty. \label{eq:lr2}
\end{align}
We also restrict ourselves to boundary conditions for which the infinite volume pressure is well defined and for which there exists a coinciding Gibbs measure. In particular, we will focus on boundary conditions of uniformly bounded density: for $\bar\rho \ge 0$, set
\[
	\mathcal{B}_{\bar{\rho}} \ :=\ \big\{\, b:\ m_q(b) \le \bar\rho\ \ \forall\, q \in \mathbb{Z}^{d} \,\big\}, \qquad m_q(b) = \#\,\big(b \cap \Delta_q\big),
\]
where $m_q(b)$ counts the points of $b$ in the unit cube $\Delta_q := q + [0,1)^d$, $q \in \mathbb{Z}^{d}$. The thermodynamic limit will be taken along the cubes $\Lambda_L := [-L, L]^{d}$, $L \ge 1$. It is briefly worth noting, since we will need to uniformly upper bound the partition function (Lemma~\ref{lem:pfbd} below), that there is a trade-off between the class of potential and the boundary conditions. That is, with super-stability,
\begin{equation}\label{eq:ss}
\text{super-stable:}\quad \exists A>0, B \ge 0 \qquad \phi(x_1,\dots,x_N) \ \ge\ A\, N^2 - B\, N 
\end{equation}
we can consider the un-bounded-density boundary conditions $\mathbb{B}^{*}_{g}$ of Procacci and Yuhjtman \cite{procacci2022classical}, treated in Appendix~\ref{app:bounds}, whereas if we restrained the class of potentials to merely be stable, then we would have to restrict the class of boundary conditions to those for which the density is uniformly bounded. The assumption of super-stability also allows an extension to Ruelle's natural tempered boundary conditions \cite{ruelle1970superstable}, those for which there exists $t > 0$ such that $\sum_{|q|_\infty \le n} m_q(b)^2 \le t\,(2n+1)^{d}$ for all integers $n \ge 0$ (with $|q_\infty|= \max_{1 \le i \le d}|q_i|$), fit naturally into this class whenever their local densities are bounded above but not uniformly bounded globally. This case is also treated in the Appendix~\ref{app:tempered}. Naturally, since both sets of boundary conditions are treated, even the superposition of a tempered boundary condition and those considered by Procacci and Yuhjtman is implicitly treated.

\begin{figure}[h!]
\centering
\includegraphics[width=\textwidth]{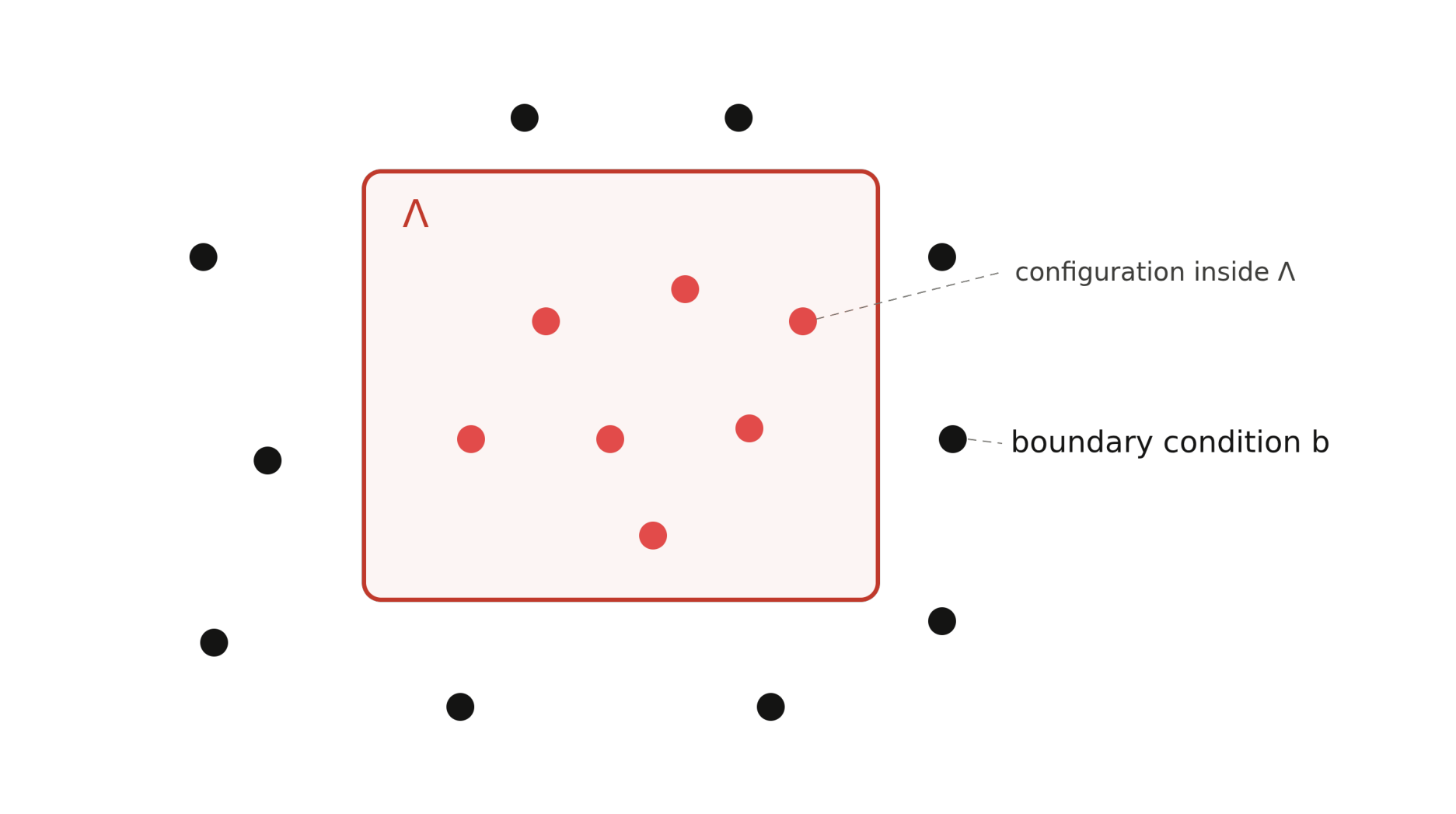}
\end{figure}

We define the finite-volume pressure $p(\mu;\Lambda,b) = |\Lambda|^{-1} \log \Xi(\mu;\Lambda,b)$, which will be the main object of our study. For appropriately chosen $b$, and for $\mu \in \mathbb{R}$ ($z > 0$), the thermodynamic pressure $p(\mu) = \lim_{|\Lambda|\to\infty} p(\mu;\Lambda,b)$, taken in the sense of Fisher or van Hove, exists, is finite, and is independent of $b$ \cite{ruelle1969statistical,ruelle1963classical,fisher1964free,friedli2017statistical}. The first derivative of the finite-volume pressure is the density,
\begin{equation}\label{eq:density}
\frac{1}{\beta} \frac{\partial p}{\partial \mu}(\mu;\Lambda,b) = \rho(\mu;\Lambda,b),
\end{equation}
and the second is the variance of the particle number per unit volume,
\begin{equation}\label{eq:variance}
\frac{1}{\beta^{2}} \frac{\partial^{2} p}{\partial \mu^{2}}(\mu;\Lambda,b) = \frac{\langle N_\Lambda^{2}\rangle - \langle N_\Lambda\rangle^{2}}{|\Lambda|}.
\end{equation}

We study the relation between $\sigma_\Lambda$ and the thermodynamic pressure $p(\mu) = \lim_{|\Lambda|\to\infty} p(\mu;\Lambda,b)$, and give conditions under which $\lim_{|\Lambda|\to\infty} \sigma_\Lambda/|\Lambda| = \beta^{-2}\partial^{2}p(\mu)/\partial\mu^{2}$. To study this limit, we must introduce the Lee--Yang zeros of the partition function, $\Xi$, in the complex fugacity plane $z = e^{\beta\mu}$. Due to the logarithmic relationship, the analytic structure of the pressure is governed by the zeros of $\Xi$. Following Lee and Yang \cite{lee1952statistical2,yang1952statistical1}, we study the analytic continuation of $\Xi(z;\Lambda,b)$ in complex fugacity $z = e^{\beta\mu}$; it has no zeros on the positive real axis, so $p(z;\Lambda,b)$ is analytic there. A phase transition occurs only where the zeros accumulate on the real axis as $|\Lambda| \to \infty$, and the manner of accumulation determines the type of singularity \cite{lebowitz1968analytic}. On any real interval that the zeros avoid, uniformly in $\Lambda$ for all large volumes, $p$ is analytic and no transition occurs. Our main result (Theorem 1.1) is that the derivatives of the pressure in $\mu$ converge and are independent of the boundary condition whenever the Lee--Yang zeros of the partition function in the complex fugacity plane $z = e^{\beta\mu}$ stay bounded away from a real point for all large $\Lambda$.

We state our result for admissible collections of boundary conditions: a collection $\mathcal{B}$ of boundary conditions is called \emph{admissible} if the grand canonical partition function is uniformly upper bounded along the cubes, in the sense that there are constants $c_1, c_2 < \infty$ and $L_1 < \infty$ with
\begin{equation}\label{eq:admissible}
	|\Lambda_L|^{-1} \log \Xi(|z|;\Lambda_L,b) \ \le\ c_1 |z| + c_2 \qquad \text{for all } z \in \mathbb{C},\ b \in \mathcal{B},\ L \ge L_1.
\end{equation}
In the main body we work with the boundary conditions $\mathcal{B}_{\bar{\rho}}$ of uniformly bounded density, only to have a simple upper bound on the partition function (Lemma~\ref{lem:pfbd}); the boundary conditions $\mathbb{B}^{*}_{g}$ of Procacci and Yuhjtman are treated in Appendix~\ref{app:bounds} and those of Ruelle's tempered boundary conditions are treated in Appenix~\ref{app:tempered}.

\begin{theorem}[Exchange of thermodynamic limit and differentiation at $z_0$]\label{thm:main}
	Let $\mathcal{B}$ be an admissible collection with common thermodynamic pressure
\[
p(x) \ =\ \lim_{L\to\infty} p(x;\Lambda_L,b), \qquad x > 0,
\]
independent of $b \in \mathcal{B}$. Fix $z_0 > 0$ and $b \in \mathcal{B}_{\bar{\rho}}$. Suppose that, for some $\delta > 0$, $\Xi(\,\cdot\,;\Lambda_L,b)$ has no zeros in $B_\delta(z_0)$ for all sufficiently large $L$. Then $p$ has a holomorphic extension to a neighborhood of $z_0$, and, for every $n \ge 1$,
\[
\lim_{L\to\infty} \partial_{z}^{n}\, p(z_0;\Lambda_L,b) \ =\ \partial_{z}^{n}\, p(z_0).
\]
Thus the thermodynamic limit and differentiation commute at $z_0$. The value on the right is independent of $b$ among boundary conditions satisfying the zero-free hypothesis at $z_0$.
\end{theorem}

In fact, the proof yields the stronger conclusion that, for every $n \ge 1$, $\partial_{z}^{n} p(\,\cdot\,;\Lambda_L,b) \to \partial_{z}^{n} p$ uniformly on $B_{\delta/2}(z_0)$. Setting $\mu_0 = \beta^{-1} \log z_0$, the finite-volume density and variance per unit volume converge at $\mu_0$ to $\beta^{-1} \partial_\mu p(\mu_0)$ and $\beta^{-2} \partial^{2}_{\mu} p(\mu_0)$, respectively.

For finite $\Lambda$ the derivatives $\partial_\mu^{n} p(\mu;\Lambda,b)$ are real-analytic in $\mu$ ($z > 0$) and depend explicitly on $b$. However, for admissible boundary conditions with stable and tempered potentials, the thermodynamic pressure is independent of $b$. Despite this, the infinite-volume Gibbs measure need not be unique: a single analytic pressure can be compatible with several Gibbs measures. An example of this would be the hard-core lattice systems at high-fugacity \cite{jauslin2018high}, in which two distinct Gibbs measures that heavily weight occupation on even or odd lattice sites correspond to the same thermodynamic pressure. The question is whether, as $|\Lambda| \to \infty$, they converge, whether the limit is independent of $b$, and whether it equals $\partial_\mu^{n} p(\mu)$, i.e. whether the thermodynamic limit commutes with differentiation. The limit cannot reproduce $\partial_\mu^{n} p(\mu)$ where $p$ is not $n$ times differentiable. Where $p$ is analytic we expect agreement but cannot prove it in general, owing to limited control over the Lee-Yang zeros. With this main result then we obtain results that extend Dereudre and Flimmel's bound for their considered systems, that is, not only are such systems non-hyperuniform but the limiting variance is equal for all coexisting phases when the Lee-Yang zeros are bounded away.

These Lee--Yang zeros are the tool the argument requires; their connection to phase transitions goes back to Lee and Yang \cite{lee1952statistical2,yang1952statistical1}. For lattice gases and equivalent ferromagnetic spin systems, the zeros in the complex fugacity plane lie on the unit circle (the circle theorem), and the pressure becomes singular precisely where they accumulate on the positive real axis in the thermodynamic limit. Fisher later introduced the analogous zeros in the complex temperature plane \cite{fisher1965nature}, and the picture of phase transitions as limiting distributions of partition-function zeros became a cornerstone of rigorous statistical mechanics.

\section{Analyticity when zeros are bounded away}\label{sec:bounded}

\noindent\textbf{Setting.}\ Let $\phi$ be stable, tempered and lower-regular. Stability gives $\phi(x) \ge -2B$, so that $\phi^{-}(x) \le \psi_0(|x|)$, where $\psi_0(r) := 2B$ for $0 \le r < R_1$ and $\psi_0(r) := \min\{2B, \psi(r)\}$ for $r \ge R_1$ is decreasing and, by \eqref{eq:lr2}, integrable; hence
\[
S \ :=\ \sum_{q \in \mathbb{Z}^{d}} \psi_0\big( (|q| - \sqrt{d}\,)^{+} \big) \ <\ \infty.
\]
where $|q|$ is the Euclidean norm of $q$ and $+$ signifies $\max(\cdot,0)$. For a boundary condition $b$ write $\phi^{-}_{b}(x) := \sum_{y \in b} \phi^{-}(x - y)$; if $\sum_{y \in b} \phi^{+}(x - y)$ diverges we set $\phi_b(x) := +\infty$ and $e^{-\beta\phi_b(x)} := 0$, so that all estimates below use only the termwise bound $\phi_b \ge -\phi^{-}_{b}$.

We begin with a lemma such that for stable, tempered and lower-regular potentials, we can uniformly upper-bound the partition function on concentric cubes:

\begin{lemma}[Partition function bound]\label{lem:pfbd}
	Let $\phi$ be stable, tempered and lower-regular, and set $K := \bar\rho\, S$. Then for every $L \ge 1$, every $b \in \mathcal{B}_{\bar{\rho}}$, and every $z \in \mathbb{C}$,
\[
|\Xi(z;\Lambda_L,b)| \ \le\ \Xi(|z|;\Lambda_L,b) \ \le\ \exp\big( |z|\, |\Lambda_L|\, e^{\beta(B+K)} \big).
\]
\end{lemma}

\begin{proof}
For $x \in \Lambda_L$, writing $q_x$ for the cell containing $x$, we have $|x - y| \ge (|q - q_x| - \sqrt{d}\,)^{+}$ for $y \in \Delta_q$, so, since $\psi_0$ is decreasing,
\[
\phi^{-}_{b}(x) \ =\ \sum_{y \in b} \phi^{-}(x - y) \ \le\ \sum_{q} m_q(b)\, \psi_0\big( (|q - q_x| - \sqrt{d}\,)^{+} \big) \ \le\ \bar\rho\, S \ =\ K,
\]
and hence $\phi_b \ge -K$ on $\Lambda_L$. For a configuration $(x_1,\dots,x_N)$ in $\Lambda_L$, stability then gives
\[
\sum_{i<j} \phi\big(|x_i - x_j|\big) + \sum_{i} \phi_b(x_i) \ \ge\ -(B + K)\, N,
\]
so each coefficient of $z^{N}$ in $\Xi(z;\Lambda_L,b)$ is non-negative and at most $\big(|\Lambda_L|\, e^{\beta(B+K)}\big)^{N}/N!$; the first inequality of the lemma is the triangle inequality applied to the series, and summing the coefficient bound gives the second.
\end{proof}

In particular $\mathcal{B}_{\bar{\rho}}$ satisfies \eqref{eq:admissible} with $c_1 = e^{\beta(B+K)}$ and $c_2 = 0$, for every $L \ge 1$, hence is admissible. From now on $p(z)$ will always refer to the limit given by
\[
p(z;\Lambda,b) = \frac{1}{|\Lambda|} \log \Xi(z;\Lambda,b) \xrightarrow{\;|\Lambda|\to\infty\;} p(z),
\]
The results of this section rest on an application of Vitali's theorem: if a region containing an interval of the positive real axis is free of Lee--Yang zeros for all sufficiently large $|\Lambda|$, then the finite-volume pressures are uniformly bounded there and converge to an analytic limit, so that no phase transition can occur in that region; moreover, all derivatives then converge uniformly, so that the thermodynamic limit commutes with differentiation in such regions.

\begin{proof}[Proof of Theorem 1.1]
	Let $L_0$ be such that $\Xi(\,\cdot\,;\Lambda_L,b)$ has no zeros in $B_\delta(z_0)$ for $L \ge L_0$, and take $L \ge \max(L_0, L_1)$, with $L_1$ as in \eqref{eq:admissible}. By admissibility, $\Xi(|z|;\Lambda_L,b) \le \exp\big(|\Lambda_L|\,(c_1 |z| + c_2)\big) < \infty$ for every $z \in \mathbb{C}$; since all the coefficients of $z$ in Eq.~\eqref{eq:partfun} are non-negative, with the $N = 0$ term equal to $1$, the series converges for every $z$, $\Xi(\,\cdot\,;\Lambda_L,b)$ is entire \cite{ruelle1969statistical}, and $\Xi(x;\Lambda_L,b) \ge 1$ for $x \ge 0$. The ball $B_\delta(z_0)$ being open, simply connected and zero-free, $|\Lambda_L|^{-1} \log \Xi(\,\cdot\,;\Lambda_L,b)$ admits a holomorphic branch there, normalized to be real at $z_0$; on $I := B_\delta(z_0) \cap \mathbb{R}_{>0}$, $\operatorname{Im} \log \Xi$ is continuous with values in $2\pi\mathbb{Z}$ and vanishes at $z_0$, hence identically, so the branch agrees on $I$ with the finite-volume pressure, and we denote it again by $p(\,\cdot\,;\Lambda_L,b)$. Since $|\Xi(z;\Lambda_L,b)| \le \Xi(|z|;\Lambda_L,b)$, \eqref{eq:admissible} gives $\operatorname{Re} p(z;\Lambda_L,b) = |\Lambda_L|^{-1} \log |\Xi(z;\Lambda_L,b)| \le c_1|z| + c_2$ on $B_\delta(z_0)$, while $0 \le p(z_0;\Lambda_L,b) \le c_1 z_0 + c_2$. The Borel--Carath\'eodory inequality \cite[\S 5.5]{titchmarsh1939theory} states that for $f$ holomorphic on $B_R(z_0)$ and $0 < r < R$,
\[
\sup_{B_r(z_0)} |f| \ \le\ \frac{2r}{R - r}\, \sup_{B_R(z_0)} \operatorname{Re} f + \frac{R + r}{R - r}\, |f(z_0)|.
\]
Fix $0 < r < R < \delta$. Applying the inequality to $p(\,\cdot\,;\Lambda_L,b)$ gives, uniformly in $L$,
\[
\sup_{B_r(z_0)} \big| p(\,\cdot\,;\Lambda_L,b) \big| \ \le\ \frac{2r}{R - r} \big( c_1 (z_0 + R) + c_2 \big) + \frac{R + r}{R - r}\, \big( c_1 z_0 + c_2 \big) \ =:\ C_{r,R}.
\]
Thus the family $\big\{ p(\,\cdot\,;\Lambda_L,b) \big\}_{L}$ is locally uniformly bounded on $B_\delta(z_0)$. Along any sequence $L_n \to \infty$, the functions $p(\,\cdot\,;\Lambda_{L_n},b)$ converge on $I$ to $p$ by hypothesis, and $I$ has the accumulation point $z_0$. Vitali's theorem \cite[\S 5.21]{titchmarsh1939theory} therefore gives local uniform convergence on $B_\delta(z_0)$ to the unique holomorphic extension of $p|_{I}$. This extension is the same for every sequence and is independent of $b$, since $p|_{I}$ is; we keep the symbol $p$ and conclude that $p(\,\cdot\,;\Lambda_L,b) \to p$ locally uniformly as $L \to \infty$.

Finally, choose $0 < r < \delta$ and $\rho > 0$ such that $r + \rho < \delta$. Cauchy's formula, applied on the circle of radius $\rho$ about each $z \in B_r(z_0)$, gives
\[
\Big| \partial_{z}^{n} \big[\, p(z;\Lambda_L,b) - p(z) \,\big] \Big| \ =\ \bigg| \frac{n!}{2\pi i} \oint_{|\zeta - z| = \rho} \frac{p(\zeta;\Lambda_L,b) - p(\zeta)}{(\zeta - z)^{n+1}} \, d\zeta \bigg| \ \le\ \frac{n!}{\rho^{n}} \sup_{B_{r+\rho}(z_0)} \big| p(\,\cdot\,;\Lambda_L,b) - p \big| \ \longrightarrow\ 0
\]
uniformly for $z \in B_r(z_0)$. Taking $z = z_0$ proves the asserted exchange of thermodynamic limit and differentiation at $z_0$. Taking $r = \delta/2$ and $\rho = \delta/4$ gives the stronger uniform convergence on $B_{\delta/2}(z_0)$ stated after the theorem.
\end{proof}

\begin{remark}[First Derivate]
Since each finite-volume pressure $p(\,\cdot\,;\Lambda,b)$ is convex in $\mu$, so is the limit $p$; this convexity is a manifestation of thermodynamic stability \cite{israel1979convexity}, and it forces the density $\rho(\mu) = \beta^{-1}\partial p/\partial\mu$ to be non-decreasing, with at most countably many points of discontinuity. Convexity also distinguishes the first derivative from the higher ones: for convex functions converging pointwise, the derivatives converge at every point where the limit is differentiable, so the finite-volume densities $\partial_\mu p(\,\cdot\,;\Lambda,b)$ converge to $\partial_\mu p$ wherever the density is single-valued---that is, away from a first-order transition---with no hypothesis on the zeros. The limit and the first derivative can be exchanged away from phase transitions, whereas the convergence of $\partial^{2}_{\mu} p$ and beyond is what requires the zeros to stay away from the real point. Relatedly, in the region of analyticity the fluctuations of the particle density $N/|\Lambda|$ about its mean are asymptotically Gaussian with variance of order $|\Lambda|^{-1}$, whereas at a first-order transition the limiting distribution becomes bimodal, reflecting the coexistence of two phases \cite{lebowitz1968analytic,jauslin2018high}.
\end{remark}

\section{Singularity-Free Regions}

Having proven the main results of the previous section, we now list some regions of the $z$-plane that are known to be free of Lee--Yang zeros (for large $|\Lambda|$) and to which classical arguments apply:
\begin{itemize}
\item[(a)] \textbf{Low fugacity (high temperature):} $|z| < R = \frac{e^{-2\beta B - 1}}{\int |e^{-\beta\phi} - 1|\, dx}$ for some convergence radius $R > 0$ that depends on the interaction. In this regime, the Mayer cluster expansion \cite{ruelle1969statistical,penrose1963convergence,mayer1940statistical,ursell1927evaluation} and the more modern approaches based on the Kirkwood--Salsburg equations \cite{gallavotti1999statistical,ruelle1963correlation} provide convergent series representations of the pressure and correlation functions. The cluster expansion yields
\begin{equation}\label{eq:mayer}
\beta p(z) = \sum_{n=1}^{\infty} b_n z^{n},
\end{equation}
where $b_n$ are the Mayer cluster integrals (or cluster coefficients). The convergence of this series was first proven rigorously by Groeneveld \cite{groeneveld1962two}, Penrose \cite{penrose1963convergence} and Ruelle \cite{ruelle1963correlation}, and has been sharpened more recently.

\item[(b)] \textbf{Ferromagnetic spin systems with magnetic field $h \ne 0$:} For Ising ferromagnets and equivalent lattice gas systems, the celebrated Lee--Yang circle theorem \cite{lee1952statistical2} states that all zeros of $\Xi$ lie on the unit circle $|z| = 1$ (equivalently, on the imaginary axis in the $\mu$-plane). Consequently, for any nonzero magnetic field (equivalently, $|z| \ne 1$), the pressure is analytic and there is no phase transition. This theorem was extended to more general ferromagnetic interactions by Griffiths \cite{griffiths1969rigorous}, Simon and Griffiths \cite{simon1973phi4}, Lieb and Sokal \cite{lieb1981general}, and Newman \cite{newman1974zeros}. More recent work by Jiang and Newman proved Lee--Yang's conjecture that the modulus of the first Lee--Yang zero decreases as a function of the system size and converges to the end-point of the Lee--Yang distribution of zeros on the imaginary axis (meaning that no zeros leak any closer to the origin) \cite{jiang2024thermodynamic}. The Lee--Yang theorem requires periodic or otherwise non-antiferro boundary conditions.

\item[(c)] \textbf{Hard-core lattice systems at large fugacity ($z \gg 1$):} For several classes of lattice gas models with extended hard-core exclusion---such as the hard-square model on $\mathbb{Z}^{2}$, the hard hexagon model, and various nearest-neighbor exclusion models on bipartite lattices---the pressure can be shown to be analytic in $z^{-1}$ for $|z|$ sufficiently large, even though the system exhibits long-range order and phase coexistence \cite{dobrushin1968problem,heilmann1972theory,baxter1980hard,pirogov1975phase,pirogov1976phase}. These systems provide the key examples for our main result: the pressure is analytic, yet there exist multiple infinite-volume Gibbs measures. In the low-temperature regime, Pirogov--Sinai theory \cite{pirogov1975phase,pirogov1976phase,zahradnik1984alternate} provides a systematic framework for establishing both the analyticity of the pressure and the multiplicity of Gibbs states in such systems. Jauslin and Lebowitz extended these results to non-sliding models, showing that there is a potential annulus cloud of Lee--Yang zeros \cite{jauslin2018high}. He and Jauslin removed the tiling constraint, extending the analysis to non-sliding models whose particles do not tile space at close packing, such as discrete hard-disk models \cite{he2024high}. The single-particle-species restriction was subsequently removed by He \cite{he2026unified}, whose criterion covers polyomino models with rotational degrees of freedom and multi-component mixtures of geometrically distinct shapes, via the Mazel--Stuhl--Suhov extension of Pirogov--Sinai theory.

\item[(d)] \textbf{Exactly solvable models:} For several exactly solvable models, such as the two-dimensional Ising model \cite{onsager1944crystal}, the hard hexagon model \cite{baxter1980hard}, and the six-vertex model \cite{lieb1972two}, the partition function and free energy are known in closed form, allowing a complete determination of the singularity structure in the complex $z$-plane.
\end{itemize}

\section{Discussion and Outlook}

Our main result ties the Lee--Yang theory of phase transitions to the convergence of arbitrary derivatives of the grand canonical partition function: wherever the Lee--Yang zeros stay bounded away from a real point as the volume grows, the finite-volume pressure and all of its derivatives in the chemical potential converge to those of the thermodynamic pressure, uniformly and independently of the boundary condition (Theorem 1.1). The thermodynamic limit there commutes with differentiation, so the limiting law of the particle number $N_\Lambda$ in a large region---and in particular its variance per unit volume $\beta^{-2}\partial^{2}p/\partial\mu^{2}$---is one and the same for every infinite-volume Gibbs measure derived from a boundary condition of uniformly bounded density. This is not obvious, since coexisting measures can differ in all their local properties: for hard-core lattice gases at large fugacity the even- and odd-sublattice phases carry distinct local densities, yet the probability of a given total particle number in a large subvolume is asymptotically equal. The total particle number is a bulk observable fixed by the pressure; the boundary condition moves the finite-volume zeros but not the limit they define, and so cannot change it.

The zero-free hypothesis of Theorem 1.1 is more targeted than it may appear. The pressure itself is independent of the boundary condition for any admissible $b$. What the hypothesis buys is the convergence of the higher derivatives, the limiting variance foremost among them: it is here that the thermodynamic limit and differentiation can fail to commute, and bounding the zeros away from the real point is what restores it.

The coexistence in the systems of Section~3(c) also constrains the zeros globally, not merely near $z_0$. Define
\[
\Xi_s(x_1,\dots,x_s;z,\Lambda,b) \ =\ \sum_{N \ge s} \frac{z^{N}}{(N-s)!} \int_{\Lambda^{N-s}} e^{-\beta\phi(x_1,\dots,x_N)} \prod_{i=1}^{N} e^{-\beta\phi_b(x_i)} \, dx_{s+1} \cdots dx_N ,
\]
The Gibbs measure there is unique at small fugacity, by the convergence of the Mayer expansion \cite{penrose1963convergence}; the finite-volume correlation functions $n_s = \Xi_s/\Xi$ are analytic in $z$ wherever $\Xi$ is zero-free; and Vitali's theorem transports the small-fugacity agreement of their limits along any zero-free region joining the origin to $z_0$, provided the $n_s$ are bounded there at complex $z$ uniformly in the volume --- a bound the pressure never requires, owing to the normalization $|\Lambda_L|^{-1}\log$. Granting that bound, two distinct Gibbs measures at $z_0$ would be impossible if a bounded connected zero-free region contained both $0$ and $z_0$, since their correlation functions, agreeing at small fugacity, would agree at $z_0$ by the identity theorem and, through the Ruelle bounds, determine the same measure \cite{ruelle1970superstable}; the coexistence therefore forces the Lee--Yang zeros to accumulate on a closed set separating $z_0$ from the origin --- for the systems of Section~3(c), to enclose the origin within the annulus of \cite{jauslin2018high}. The required bound on the $n_s$ is established in \cite{lebowitz1968analytic} for lattice gases with attractive pair interactions, through the Lee--Yang factorization, and for non-negative potentials within the zero-free disc about the origin, by Groeneveld's identity \cite{groeneveld1962two}; in the first class coexistence occurs only on the zero set itself, and in the second the covered region is one of low-activity uniqueness. In the intermediate region between the low- and high-fugacity expansions, where the statement has its content, the bound is not known, so the enclosure of the origin, while consistent with \cite{jauslin2018high}, remains a conjecture; establishing it, or working instead with the pressure of a two-activity ensemble carrying a sublattice symmetry-breaking field in the spirit of \cite[\S IX]{lebowitz1968analytic}, would make it a theorem.

The deeper limitation is not physical but methodological: Theorem 1.1 takes the location of the Lee--Yang zeros as a hypothesis rather than supplying it. We invoke a zero-free region, but for a given boundary condition we have no a priori control of where the zeros sit, and it is precisely this that confines our conclusions to the regimes of Section 3, where zero-freeness is established by independent means---cluster expansions, the circle theorem, or Pirogov--Sinai theory. We expect the convergence of all derivatives, and their boundary-condition independence, to hold at every point at which $p$ is analytic, irrespective of how the finite-volume zeros approach the real axis; a finite number of zeros drifting toward a real point at which $p$ remains smooth ought not to matter. Proving this would require quantitative bounds on the zeros---on their density near the real axis and the rate of their approach---of a kind not currently available away from the special models above, and sharpening such bounds, in the spirit of the recent control of the first Lee--Yang zero for ferromagnets \cite{jiang2024thermodynamic}, is the natural next step. Beyond it lie the regimes the present hypotheses exclude by construction: genuinely hyperuniform states, where the volume coefficient $\partial^{2}p/\partial\mu^{2}$ vanishes and the variance grows only sub-extensively, demanding either the long-range interactions of Coulomb systems or the vanishing compressibility of a critical point where the zeros reach the real axis; and the full particle-number distribution at a transition, where the law common to the coexisting phases away from it must break down.

\appendix

\section{The boundary conditions of Procacci and Yuhjtman}\label{app:bounds}

In this appendix we verify \eqref{eq:admissible} for the boundary conditions $\mathbb{B}^{*}_{g}$ of Procacci and Yuhjtman \cite{procacci2022classical}, of possibly unbounded density, at the price of strengthening stability to super-stability. Let $\phi$ be super-stable, tempered and lower-regular in the sense of \eqref{eq:temp}--\eqref{eq:ss}, with super-stability constants $A > 0$, $B \ge 0$ and lower-regularity envelope $\psi$, and let $\psi_0$ be as in Section~\ref{sec:bounded}. Let $g\colon [0,\infty) \to [0,\infty)$ be continuous, non-decreasing and subadditive, define
\[
\mathbb{B}^{*}_{g} \ :=\ \bigcup_{\bar\rho \ge 0} \Big\{\, b:\ m_q(b) \le \bar\rho\, \big(1 + g(|q|_\infty)\big)\ \ \forall\, q \in \mathbb{Z}^{d} \,\Big\}, \qquad m_q(b) = \#\,\big(b \cap \Delta_q\big),
\]
and assume the growth conditions
\[
\lim_{L\to\infty} \frac{\big(1 + g(L)\big)^{2}}{L} \int_{0}^{L} V(r)\, dr \ =\ 0, \qquad
\lim_{L\to\infty} \frac{1 + g(L)}{L} \int_{0}^{L} W(r)\, dr \ =\ 0,
\]
where $V(r) := \int_{|x| \ge r} \psi_0(|x|)\, dx$ and $W(r) := \int_{|x| \ge r} \psi_0(|x|)\, g(|x|)\, dx$. The growth conditions directly imply that  $\int_{\mathbb{R}^{d}} \psi_0(|x|)\big(1 + g(|x|)\big)\, dx < \infty$.  (If $\psi_0(r) \le Cr^{-(d+p)}$ for large $r$ and some $p > 0$, then $g(r) = r^{q}$ with $q < \tfrac12 \min\{1, p\}$ is admissible.)

For $b \in \mathbb{B}^{*}_{g}$, since $\psi_0$ is decreasing and integrable against $1 + g$, and $g$ is subadditive, $\phi^{-}_{b}(x) \le \bar\rho \sum_{q} \big(1 + g(|q|_\infty)\big)\, \psi_0\big((|q - q_x| - \sqrt{d}\,)^{+}\big) < \infty$, locally uniformly in $x$.

Along the cubes $\Lambda_L := [-L, L]^{d}$, $L \ge 1$, set, for each unit cell $\Delta_q$ meeting $\Lambda_L$,
\[
s_q \ :=\ \sup_{\Delta_q \cap \Lambda_L} \phi^{-}_{b}, \qquad S_L \ :=\ \sum_{q} s_q, \qquad K_L \ :=\ \max_{q} s_q \ \ (< \infty \ \text{by the above}).
\]
By \cite[Lemma~4.1]{procacci2022classical}, whose hypotheses are exactly the growth conditions above,
\[
S_L K_L \ =\ o\big(|\Lambda_L|\big) \qquad (L \to \infty);
\]
in the notation of \cite{procacci2022classical} this is the statement $E_\Lambda/|\Lambda_L| \to 0$ for $E_\Lambda := (S_L K_L)^{1/3} |\Lambda_L|^{2/3}$.

\medskip

The analogue of Lemma~\ref{lem:pfbd} for these boundary conditions is:

\begin{lemma}[Partition function bound for $\mathbb{B}^{*}_{g}$]\label{lem:pf}
For every $L \ge 1$, every $b \in \mathbb{B}^{*}_{g}$, and every $z \in \mathbb{C}$,
\[
|\Xi(z;\Lambda_L,b)| \ \le\ \Xi(|z|;\Lambda_L,b) \ \le\ \exp\Big( |z|\, |\Lambda_L|\, e^{\beta B} + \frac{\beta\, S_L K_L}{4A} \Big);
\]
in particular there is $L_1 < \infty$ such that $|\Lambda_L|^{-1} \log \Xi(|z|;\Lambda_L,b) \le |z| e^{\beta B} + 1$ for all $L \ge L_1$ and $z \in \mathbb{C}$.
\end{lemma}

\begin{proof}
Let $(x_1,\dots,x_N)$ be a finite configuration in $\Lambda_L$ with $N_q$ denoting the number of the $x_i$ in the unit cube $q + [0,1)^d$. and $t := \sum_{q} N_q^{2}$. By Cauchy--Schwarz and $s_q \le K_L$,
\[
\sum_{i} \phi^{-}_{b}(x_i) \ \le\ \sum_{q} N_q s_q \ \le\ t^{1/2} \Big( \sum_{q} s_q^{2} \Big)^{1/2} \ \le\ \big( t\, S_L K_L \big)^{1/2},
\]
while \eqref{eq:ss} gives $\sum_{i<j} \phi \ge At - BN$; since $Au - (u S_L K_L)^{1/2} \ge -S_L K_L/(4A)$ for all $u \ge 0$,
\[
\sum_{i<j} \phi\big(|x_i - x_j|\big) + \sum_{i} \phi_b(x_i) \ \ge\ -BN - \frac{S_L K_L}{4A}.
\]
Hence each coefficient of $z^{N}$ in $\Xi(z;\Lambda_L,b)$ is non-negative and at most $(|\Lambda_L| e^{\beta B})^{N} e^{\beta S_L K_L/(4A)}/N!$; the first inequality is the triangle inequality applied to the series, summing the coefficient bound gives the second, and the last claim follows from $S_L K_L = o(|\Lambda_L|)$.
\end{proof}

In particular $\mathbb{B}^{*}_{g}$ satisfies \eqref{eq:admissible} with $c_1 = e^{\beta B}$ and $c_2 = 1$, hence is admissible. Moreover, by \cite{procacci2022classical}, the thermodynamic limit $p(x) = \lim_{L\to\infty} p(x;\Lambda_L,b)$, $x > 0$, exists along the cubes and is independent of $b \in \mathbb{B}^{*}_{g}$, so Theorem~\ref{thm:main} applies to $\mathbb{B}^{*}_{g}$.

\section{Ruelle's tempered boundary conditions}\label{app:tempered}

In this appendix we verify \eqref{eq:admissible} for the tempered boundary conditions of Ruelle \cite{ruelle1970superstable}, again at the price of strengthening stability to super-stability. Let $\phi$ be super-stable, tempered and lower-regular in the sense of \eqref{eq:temp}--\eqref{eq:ss}, with super-stability constants $A > 0$, $B \ge 0$ and lower-regularity envelope $\psi$, let $\psi_0$ and $S$ be as in Section~\ref{sec:bounded}, and let $b$ be tempered with constant $t$, that is, $\sum_{|q|_\infty \le n} m_q(b)^2 \le t\,(2n+1)^{d}$ for all integers $n \ge 0$.

\begin{lemma}[Partition function bound for tempered boundary conditions]\label{lem:pftemp}
Let $\phi$ be super-stable and lower-regular, and let $b$ be tempered with constant $t$. Then there is a constant $C_0 = C_0(d, R_1, \psi, B, t)$ such that, for every $L \ge 1$ and every $z \in \mathbb{C}$,
\[
|\Xi(z;\Lambda_L,b)| \ \le\ \Xi(|z|;\Lambda_L,b) \ \le\ \exp\Big( |\Lambda_L| \Big( |z| + \frac{\beta\, (B + C_0)^{2}}{4A} \Big) \Big).
\]
\end{lemma}

\begin{proof}
The proof rests on the following estimate on the boundary field.

\medskip
\noindent\textit{Claim.} There is $C_0 = C_0(d, R_1, \psi, B, t) < \infty$ such that for every $L \ge 1$,
\[
\frac{1}{|\Lambda_L|} \sum_{q' :\, \Delta_{q'} \cap \Lambda_L \ne \emptyset} \Big( \sup_{x \in \Delta_{q'}} \phi^{-}_{b}(x) \Big)^{2} \ \le\ C_0^{2}.
\]

Granting the claim, the lemma follows in three lines. For a configuration $(x_1, \dots, x_N)$ in $\Lambda_L$, with $N_{q'}$ denoting the number of the $x_i$ in $\Delta_{q'}$, super-stability \eqref{eq:ss} and the claim's integrand bound the total energy by
\[
\phi(x_1, \dots, x_N) + \sum_{i=1}^{N} \phi_b(x_i) \ \ge\ \sum_{q'} \Big( A N_{q'}^{2} - (B + s_{q'})\, N_{q'} \Big) \ \ge\ -\frac{1}{4A} \sum_{q'} \big( B + s_{q'} \big)^{2},
\]
where $s_{q'}$ denotes the supremum in the claim, and the second step is the elementary inequality $A N^{2} - c N \ge -c^{2}/4A$. By the claim, $\sum_{q'} (B + s_{q'})^{2} \le |\Lambda_L| (B + C_0)^{2}$ up to the harmless replacement of the cube count by $|\Lambda_L|$, and inserting the energy bound into the partition function gives the assertion, exactly as in the proof of Lemma~\ref{lem:pfbd}.

It remains to prove the claim. Fix $q'$ with $\Delta_{q'} \cap \Lambda_L \ne \emptyset$, so $|q'|_\infty \le L + 1$. Since points of $b$ in $\Delta_q$ are at distance at least $(|q - q'| - \sqrt{d}\,)^{+}$ from any $x \in \Delta_{q'}$, and $\phi^{-} \le \psi_0(|\cdot|)$ with $\psi_0$ decreasing (Section~\ref{sec:bounded}),
\[
s_{q'} \ \le\ \sum_{q \in \mathbb{Z}^{d}} m_q(b)\, \psi_0\big( (|q - q'| - \sqrt{d}\,)^{+} \big),
\]
where, as in Section~\ref{sec:bounded}, $S = \sum_{q \in \mathbb{Z}^{d}} \psi_0\big( (|q| - \sqrt{d}\,)^{+} \big) < \infty$. Split the sum at $|q|_\infty \le 4L$ and write $s_{q'} \le s'_{q'} + s''_{q'}$ accordingly.

For the inner part, Young's inequality for the convolution of $\psi_0\big( (|\cdot| - \sqrt{d}\,)^{+} \big) \in \ell^{1}(\mathbb{Z}^{d})$ with $m(b)\, \mathbf{1}_{\{|q|_\infty \le 4L\}} \in \ell^{2}(\mathbb{Z}^{d})$, and then the tempering hypothesis at radius $4L$, give
\[
\sum_{q'} \big( s'_{q'} \big)^{2} \ \le\ S^{2} \sum_{|q|_\infty \le 4L} m_q(b)^{2} \ \le\ S^{2}\, t\, (8L + 1)^{d} \ \le\ 8^{d} S^{2} t\, |\Lambda_L|.
\]
For the outer part, $|q|_\infty > 4L$ while $|q'|_\infty \le L + 1 \le 2L$, so $|q - q'| \ge |q|_\infty/2$, giving the bound, uniform in $q'$,
\[
s''_{q'} \ \le\ \sum_{|q|_\infty > 4L} m_q(b)\, \psi_0\big( (|q|_\infty/2 - \sqrt{d}\,)^{+} \big) \ =\ \sum_{n > 4L} a_n\, \psi_0\big( (n/2 - \sqrt{d}\,)^{+} \big), \qquad a_n := \sum_{|q|_\infty = n} m_q(b).
\]
By Cauchy--Schwarz on each sphere $\{|q|_\infty = n\}$, which contains at most $2d(2n+1)^{d-1}$ points, and by the tempering hypothesis,
\[
A_n \ :=\ \sum_{k \le n} a_k \ \le\ \Big( (2n+1)^{d} \sum_{|q|_\infty \le n} m_q(b)^{2} \Big)^{1/2} \ \le\ t^{1/2} (2n+1)^{d}.
\]
Summation by parts against the decreasing $n \mapsto \psi_0\big( (n/2 - \sqrt{d}\,)^{+} \big)$ then gives
\begin{align*}
\sum_{n > 4L} a_n\, \psi_0\big( (n/2 - \sqrt{d}\,)^{+} \big) \ &\le\ \sum_{n > 4L} A_n \Big( \psi_0\big( (n/2 - \sqrt{d}\,)^{+} \big) - \psi_0\big( ((n+1)/2 - \sqrt{d}\,)^{+} \big) \Big) \\
&\le\ t^{1/2} \sum_{n \ge 1} (2n+1)^{d} \Big( \psi_0\big( (n/2 - \sqrt{d}\,)^{+} \big) - \psi_0\big( ((n+1)/2 - \sqrt{d}\,)^{+} \big) \Big),
\end{align*}
and a second summation by parts converts the right-hand side into $c_d\, t^{1/2} \sum_{n \ge 1} (2n+1)^{d-1}\, \psi_0\big( (n/2 - \sqrt{d}\,)^{+} \big)$, which is finite by \eqref{eq:lr2} and comparable to $S$. Hence $s''_{q'} \le c_d\, t^{1/2} S$ uniformly in $q'$, and
\[
\sum_{q'} \big( s''_{q'} \big)^{2} \ \le\ c_d^{2}\, t\, S^{2}\, |\Lambda_L|.
\]
Combining the two parts and setting $C_0^{2} := 2\big( 8^{d} + c_d^{2} \big)\, t\, S^{2}$ proves the claim.
\end{proof}

In particular, for each $t > 0$, the collection of boundary conditions tempered with constant $t$ satisfies \eqref{eq:admissible} with $c_1 = 1$ and $c_2 = \beta\, (B + C_0)^{2}/(4A)$, for every $L \ge 1$, hence is admissible. Theorem~\ref{thm:main} applies whenever its common thermodynamic-limit hypothesis also holds.

\bibliographystyle{unsrt}
\bibliography{bibliography}

\end{document}